\documentclass{elsarticle}
\usepackage{amssymb,filecontents,amsbsy,amsmath,amsthm}

\def\myN{N}
\def\RHPi{\mathrm{\#RH}\Pi_1}
\def\alphabet{\Sigma}
\def\SAT{\textsc{\#Sat}}

\def\APeq{\equiv_\mathrm{AP}}
\def\APred{\leq_\mathrm{AP}}
\def\calE{\mathcal{E}}

\def\calV{\mathcal{V}}

\def\uhTutte{\textsc{UniformHyper}\allowbreak\textsc{Tutte}}
\def\BIS{\textsc{\#BIS}}

\def\ZPotts{Z_\mathrm{Potts}}
\def\ZTutte{Z_\mathrm{Tutte}} 
\def\ZtildeTutte{\widetilde{Z}_\mathrm{Tutte}}

\def\NP{\mathrm{NP}}
\def\RP{\mathrm{RP}}
\def\poly{\mathop{\mathrm{poly}}}
\let\epsilon=\varepsilon

\def\terms{t}

\def\graph{G}
  
\def\graphvertices{V}
\def\graphedges{E}
\def\hypergraph{H}
\def\hypervertices{\calV}
\def\hyperedges{\calE}
\def\hyperedge{f}
\def\hypervertex{v}

\def\boldgamma{\boldsymbol{\gamma}}  
\def\boldone{\boldsymbol{1}}
 
 \def\N{\mathbb{N}}
\def\gf2{\mathrm{GF}(2)}
\def\gfq{\mathrm{GF}(q)}
\def\matrix{M}  
\def\field{F}
\def\columns{E}  
\def\numrows{|\rows|}
\def\rows{V}
\def\subsetcols{A}  
\def\matroid{\mathcal{M}}
\def\calM{\mathcal{M}}
\def\rank{r}
\def\column{e}
\def\row{i}
\def\binTutte{\textsc{Binary}\allowbreak\textsc{Matroid}\allowbreak\textsc{Tutte}}
\def\varybinTutte{\textsc{VarBinary}\allowbreak\textsc{Matroid}\allowbreak\textsc{Tutte}}
\def\ZIsing{Z_\mathrm{Ising}}
\def\lineq{\Lambda}
\def\mono{\mathrm{mono}}
\def\goodcol{\mathrm{sat}}

\def\u3hTutte{\textsc{$3$-UniformHyperTutte}}

\def\group{\Gamma}  
\def\points{\nu}
\def\perm{g}

\def\cycles{\mathrm{cyc}}
\def\ZCycleIndex{Z_\mathrm{CI}}
\def\CycleIndex{\textsc{CycleIndex}}
\def\ciparam{x}

\def\WE{\textsc{WE}}

\newtheorem{theorem}{Theorem}
\newtheorem{lemma}[theorem]{Lemma}
\newtheorem{corollary}[theorem]{Corollary}

\begin{document}

\title{Approximating 
the Tutte polynomial of a binary matroid and other related combinatorial polynomials\tnoteref{t1}}
\tnotetext[t1]{This work was partially supported by 
the grant ``Computational Counting'', funded by the \emph{Engineering and Physical Sciences
Research Council}.}

\author[lag]{Leslie Ann Goldberg}

\address[lag]{Department of Computer Science,
University of Liverpool, Ashton Building,
Liverpool L69 3BX, United Kingdom.}

\author[mj]{Mark Jerrum}

\address[mj]{School of Mathematical Sciences\\
Queen Mary, University of London, Mile End Road, London E1 4NS, United Kingdom.}

\begin{abstract}
We consider the problem of approximating certain combinatorial polynomials. 
First, we consider the problem  of approximating the Tutte polynomial of a binary matroid with parameters~$q\geq 2$ and~$\gamma$.  (Relative to the classical $(x,y)$ parameterisation, $q=(x-1)(y-1)$ and $\gamma=y-1$.) 
A graph is a special case of a binary matroid, so earlier work by the authors shows inapproximability (subject to certain  complexity assumptions) for $q>2$, apart from the trivial case $\gamma=0$. 
The situation for $q=2$ is different. 
Previous results for graphs imply inapproximability in the region $-2\leq\gamma<0$, 
apart from at two ``special points'' where the polynomial can be computed exactly in polynomial time. 
For binary matroids, we extend this result by showing
(i)~there is no FPRAS in the region $\gamma<-2$ unless $\mathrm{NP}=\mathrm{RP}$,
and (ii)~in the region $\gamma>0$, the approximation problem is hard for the complexity class $\mathrm{\#RH}\Pi_1$ 
under approximation-preserving (AP) reducibility. 
The latter result indicates a gap in approximation complexity at $q=2$:  
whereas an FPRAS is known in the graphical case, there can be 
none in the binary matroid case, unless there is an FPRAS for all of $\mathrm{\#RH}\Pi_1$. 
The result also implies that it is computationally difficult to approximate the  weight enumerator 
of a  binary linear code, apart from at the special weights at which the problem is exactly solvable in polynomial time. 
As a consequence, we show that approximating the cycle index polynomial of a permutation group is hard 
for $\mathrm{\#RH}\Pi_1$ under AP-reducibility, partially resolving a question that we first posed in~1992.
\end{abstract}

\begin{keyword}Tutte polynomial of a binary matroid\sep weight enumerator of a binary linear code\sep
cycle index polynomial.
\end{keyword}

\maketitle
 
\section{Introduction}
 
 The
multivariate Tutte polynomial (in $q^{-1}$ and~$\boldgamma$)
of a matroid~$\matroid$ with ground set~$\columns$ and rank function $\rank_{\matroid}$
is 
defined as follows (see~\cite[(1.3)]{Sokal05})
\begin{equation}
\label{tildedef}
\ZtildeTutte(\matroid;q,\boldgamma)=\sum_{\subsetcols\subseteq\columns}
q^{-\rank_{\matroid}(\subsetcols)}
\prod_{\column\in\subsetcols}
\gamma_\column,
\end{equation}
where 
$\boldgamma=\{\gamma_\column\}_{\column\in\columns}$.
The tilde in $\ZtildeTutte$ is for consistency with~\cite{Sokal05};
in general we follow the notation used there, so we can conveniently
access various useful identities.

An important class of matroids are the graphic matroids, i.e., those that 
arise as the cycle matroid $\matroid(G)$ of some graph~$G$
(see Section~\ref{sec:binmat} for details).  The Tutte polynomial of graphic 
matroids in particular has received much attention.
For convenience, we offen speak about the
Tutte polynomial of the graph~$G$ rather than the Tutte polynomial of its cycle 
matroid $\matroid(G)$.  Indeed, this polynomial was first defined for graphs, 
and only later generalised to matroids.  

The Tutte polynomial encodes a large quantity of combinatorial information about the 
matroid~\cite{JVW90, Sokal05, welsh}
and the complexity of computing the polynomial has been 
much studied~\cite{thor, tuttepaper, planarTutte, FerroPotts, JVW90,  vertigan}.
One important motivation for this study of the Tutte polynomial is that it includes
as a special 
case (when $q$ is a positive integer)
the problem of computing the partition function of the Potts model 
in statistical physics.
To be more precise about the computational task, 
parameters~$q$ and~$\gamma$ are fixed,
and  the  problem is
to compute $\ZtildeTutte(\matroid;q,\boldgamma)$
for an input matroid~$\matroid$, where~$\boldgamma$ is
the constant function  with $\gamma_e = \gamma$ for every ground set element~$e$.
In order to introduce the topic, we assume in this introduction that the parameters~$q$ and~$\gamma$
are rational, though we shall see below that this can be generalised.

Jaeger, Vertigan and Welsh~\cite{JVW90} investigated the complexity of exactly 
computing~$\ZtildeTutte(\matroid;q,\boldgamma)$ given an input matroid~$\matroid$.
They used a different parameterisation of the polynomial, but the problem that they studied is identical
to the one that we describe here.
For the record, 
they~\cite[(2.2)]{JVW90} define
\begin{equation}
\label{TTutteDef}
T(\matroid;x,y) = \sum_{A\subseteq \columns}
{(x-1)}^{\rank_{\matroid}(\columns)-\rank_{\matroid}(A)} {(y-1)}^{|A|-\rank_{\matroid}(A)}.
\end{equation}
Now substitute  $\gamma = y-1$ and $q=(x-1)(y-1)$.
Let $\boldgamma$ be the constant function  with 
$\boldgamma_\column = \gamma$ for every 
$\column\in\columns$.
Assuming that $q\neq 0$ (hence $\gamma\not=0$), Equation~(\ref{tildedef}) gives
\begin{equation}
\label{eqtrans}
T(\matroid;x,y)=(x-1)^{\rank_{\matroid}(\columns)}\ZtildeTutte(\matroid;(x-1)(y-1),\boldgamma)
=(q/\gamma)^{\rank_{\matroid}(\columns)}\ZtildeTutte(\matroid;q,\gamma).
\end{equation}

Unfortunately, any two-parameter version of the Tutte polynomial
will  omit some points.  
On the one hand, setting $y=1$ in (\ref{TTutteDef}) forces $q=(x-1)(y-1)$ to be~$0$
but setting $\gamma_\column=\gamma=y-1=0$ in (\ref{tildedef}) does not force~$q$ to be~$0$.
On the other hand, the single point $q=0,\gamma=0$ in (\ref{tildedef}) 
corresponds to an entire line 
in the $(x,y)$-coordinate system, where $y=1$ but $x$ can have any value. 
For this reason, it is sometimes convenient \cite[\S2.3]{Sokal05}
to treat the $q\rightarrow 0$ case as a limit
case. We will not need to do this here.

Jaeger et al.\ showed that, even when the input is restricted to be a \emph{graphic} matroid,
exact evaluation is \#P-hard, apart from when $q=1$ and
at four ``special points''. 
The first three of these are 
$(q,\gamma)=(4,-2)$, $(2,-2)$ and $(2,-1)$.
The fourth is the point $(x,y)=(1,1)$ for which Equation~(\ref{eqtrans}) is invalid due to division by~$0$ ---
evaluation at this point corresponding to counting spanning trees in the input graph.
As already noted, the line $\gamma=0$ is also easy in our parameterisation.
Thus, exactly evaluating the Tutte polynomial of a graph is \#P-hard, apart from
when $q=1$ and
at these special points. Jaeger et al.\ also considered the case in which the parameters are complex
numbers, where there are additional special points, but we do not consider this case here.
As they noted, exact evaluation can be done in polynomial time for $q=1$
and it can be done in polynomial time at some of 
the special points for large classes of matroids.
For binary matroids, which are a focus of this paper, the first three special points at least are 
polynomial-time computable. A~definition of binary matroid is given in Section~\ref{sec:binmat}.

Our earlier work~\cite{tuttepaper, FerroPotts} investigates the complexity of
\emph{approximately} computing~$\ZtildeTutte(\matroid;q,\boldgamma)$
when $\matroid$ is restricted to be graphic.
We are interested in determining for which points $(q,\gamma)$ there is a \emph{fully polynomial
randomised approximation scheme} (FPRAS) for the polynomial. An FPRAS is a polynomial-time randomised
approximation algorithm achieving arbitrarily small relative error. A precise definition is provided in Section~\ref{sec:defs}.  We survey the main results now, partly because we build on them
in this article, and partly to highlight the differences in computational complexity
between the graphic and binary cases.
 
For $q>2$ we gave inapproximability results both for~$\gamma<0$ and for~$\gamma>0$.
As already noted, the case $\gamma=0$ is trivial.
In the ``antiferromagnetic'' case $\gamma<0$, we showed~\cite{tuttepaper} that, apart from at the special
point $(q=4,\gamma=-2)$,
there is no FPRAS for approximately evaluating the Tutte polynomial of a graph unless 
$\NP=\RP$.
In the ``ferromagnetic'' case $\gamma>0$, we showed~\cite{FerroPotts} 
that the approximation problem is hard for the logically-defined complexity class $\RHPi$ under 
approximation-preserving ``AP-reductions''.

The complexity class $\RHPi$ 
of counting problems was introduced by 
Dyer,  Goldberg, Greenhill and Jerrum~\cite{APred} as a means  
to classify a wide class of approximate counting problems that were
previously of indeterminate computational complexity.  The problems in   $\RHPi$
are those that 
can be expressed in terms of counting the number of models of a logical formula
from a certain syntactically restricted class which is also known as  ``restricted Krom SNP''~\cite{Dalmau05}.
  $\RHPi$ has a completeness class  with respect to AP-reductions which includes a wide  
range of natural counting problems --- see Section~\ref{sec:defs} for some examples.
Either all of these problems admit an FPRAS,
or none do.  No FPRAS is known for any of them at the time of writing, despite 
much effort having been expended on finding one.  We conjecture that none exists.
Proving counting problems to be hard for $\RHPi$ with respect to AP-reductions is 
similar to working
with the Unique Games Conjecture in the area of approximation algorithms
for optimisation problems, or employing the class PPAD in analysing
the complexity of Nash equilibria.
Since a graphical matroid is a binary matroid, both of the hardness results 
for $q>2$ mentioned earlier 
(for $\gamma<0$ and for $\gamma>0$)
extend to the class of binary matroids.

The paper~\cite{tuttepaper} also includes hardness results for $q<2$
which extend to the binary matroid case.  For example, there is no FPRAS unless $\NP=\RP$ 
if either $\gamma$ or $q/\gamma$ is less than~$-2$. The interested reader is referred to~\cite{tuttepaper}.

The situation is different for $q=2$.
In this case, we showed~\cite{tuttepaper}
that in the region $-2<\gamma<0$ (apart from at special points)
there is no FPRAS for approximately evaluating the Tutte polynomial of a graph unless 
$\NP=\RP$. However,
the most that is known for $\gamma<-2$ (see~\cite{tuttepaper}) is that the problem is as difficult as approximately counting
perfect matchings in a graph, a well-known open problem.
For $\gamma>0$, the Tutte polynomial of a graph \emph{can} be approximated efficiently --- 
Jerrum and Sinclair
have given an FPRAS~\cite{JS93}.
In this paper, we show that the problem of approximating the Tutte polynomial of a binary matroid
is apparently harder.
In particular, we show in Theorem~\ref{thm:binmatroid} that there is no FPRAS 
in the region $\gamma<-2$
unless $\NP=\RP$
and that the problem is hard for $\RHPi$ with respect to AP-reductions for $\gamma>0$.

The results in Theorem~\ref{thm:binmatroid} have
interesting consequences for the problem of approximating related polynomials.
It is well-known that the Tutte polynomial of a binary matroid 
contains as a special case
the weight enumerator
of a binary linear code, which will be defined in Section~\ref{sec:we}. Hence, we immediately get a complexity classification (Corollary~\ref{cor:we}) for this problem.
This, in turn, allows us to make progress on a long-standing open problem about the complexity
of approximating the cycle index polynomial of a permutation group (see Section~\ref{sec:ci}).
We had previously shown~\cite{CIpaper} that there is no FPRAS for this problem, unless 
$\NP=\RP$, if the parameter, $\ciparam$,
is a non-integer.
Using our result about the weight enumerator of a  binary linear code, we  show that the cycle index
polynomial is  as difficult to approximate as $\RHPi$
 for every parameter value $\ciparam>1$ (Corollary~\ref{cor:ci}).
As we will explain in Section~\ref{sec:ci}, it is at the integer points that the cycle index polynomial has combinatorial meaning.

\subsection{Matroid preliminaries}\label{sec:matroid}
A matroid $\calM$ is a combinatorial structure  defined  by a 
set $E$ (the ``ground set'') 
together with a ``rank function'' $r_{\calM}:E\to\N$. The
rank function satisfies the following conditions  for
all subsets $A,B\subseteq E$: 
(i)~$0\leq r_{\calM}(A)\leq |A|$, 
(ii)~$A\subseteq B$ implies $r_{\calM}(A)\leq r_{\calM}(B)$ (monotonicity), 
and (iii)~$r_{\calM}(A\cup B)+r_{\calM}(A\cap B)
\leq r_{\calM}(A)+r_{\calM}(B)$ (submodularity).  

A subset $A\subseteq E$ satisfying $r_{\calM}(A)=|A|$
is said to be {\it independent}.
Every other subset $A\subseteq E$ is said to be {\it dependent.}
A  maximal (with respect to inclusion)
independent set is   a {\it basis}, and a minimal dependent set is   a {\it circuit}.
A~circuit with one element is  a {\it loop}.

Suppose that $\calM$   is a matroid with ground set~$E$.
Then $\calM$ is associated with  a {\it dual matroid\/}~$\calM^*$
with the same ground set $E$ but rank function $r_{\calM^*}$ 
given by $r_{\calM^*}(A)=|A|+r_{\calM}(E-A)-r_{\calM}(E)$.
A {\it cocircuit\/} in $\calM$ is a set that is a circuit in~$\calM^*$;
equivalently, a cocircuit is a minimal set that intersects every basis.
A~cocircuit with one element is a {\it coloop}.

We will use the matroid operations \emph{contraction} and \emph{deletion}. Suppose 
$e \in E $ is a member of the ground set of matroid~$\calM$.  
The {\it contraction
$\calM/e$ of $e$ from $\calM$} is the matroid on ground set $E-\{e\}$ with 
rank function given by $r_{\calM/e}(A)=r_\calM(A\cup \{e\})-r_\calM(\{e\})$, for
all $A\subseteq E-\{e\}$.
The {\it deletion
$\calM\backslash e$ of~$\{e\}$ from $\calM$} is the matroid on ground set $E-\{e\}$ with 
rank function given by $r_{\calM \backslash e}(A)=r_\calM(A)$, for
all $A\subseteq E-\{e\}$.   
We refer the reader to Oxley's book~\cite{oxley}
for a thorough exposition of   matroid theory.

\subsection{The Tutte polynomial of a binary matroid}
\label{sec:binmat}

Let $\matrix$ be a matrix over a field $\field$ with 
row set~$\rows$ and column set~$\columns$.
$\matrix$ ``represents'' a matroid~$\matroid$ with ground set $\columns$.
The rank  $\rank_{\matroid}(\subsetcols)$ of a set of columns~$\subsetcols$
in the matroid
is defined to be the rank of the submatrix consisting of those columns.  It is easy
to see (see \cite{oxley}) that a rank function defined in this way satisfies the three conditions 
(i)--(iii) for a matroid rank function presented in the previous subsection.
Therefore, 
a set $\subsetcols\subseteq \columns$ is
dependent in the matroid
if and only if  the columns in~$\subsetcols$ are linearly dependent as vectors.
A matroid is said to be \emph{representable} over the field~$\field$ if it can be represented in this way.
It is said to be \emph{binary} if  it is representable over~$\gf2$.

The \emph{cycle matroid} of an undirected graph $\graph=(\graphvertices,\graphedges)$
is the matroid $\matroid(\graph)$ represented by the vertex-edge
incidence matrix~$\matrix$ of~$\graph$.  In this case,
$\rank_{\matroid(\graph)}(\subsetcols) = |\graphvertices| - \kappa(\graphvertices,\subsetcols)$
where $\kappa(\graphvertices,\subsetcols)$ is the number of connected components of
the graph $(\graphvertices,\subsetcols)$.
We simplify notation by writing $\ZtildeTutte(\graph;q,\boldgamma)$
in place of $\ZtildeTutte(\matroid(\graph);q,\boldgamma)$.  
Since the Tutte polynomial of
a binary matroid generalises the Tutte polynomial of a graph, any hardness result for the 
latter immediately translates to the former.   In this context, it should be noted that 
there is a slight mismatch between the definition of the Tutte polynomial 
given here, and the one used in the papers we cite, e.g., \cite{FerroPotts, Sokal05}.
There, the Tutte polynomial of a graph is defined using the ``random cluster'' formulation:
\begin{equation}\label{eq:randClust}
\ZTutte(\graph;q,\boldgamma)=\sum_{\subsetcols\subseteq E}
q^{\kappa(\graphvertices,\subsetcols)}
\prod_{\column\in\subsetcols}
\gamma_\column=q^{|V|}\,\ZtildeTutte(\graph;q,\boldgamma).
\end{equation}
(Note the absence of a tilde!)
But since the two formulations differ only by an easily-computable factor $q^{|V|}$,
all complexity results, whether about approximate or exact computation,
translate directly.

For fixed real numbers $q$ and $\gamma$  we define
the following computational problem, which is parameterised by~$q$ and~$\gamma$.

\begin{description}
\item[Problem] $\binTutte(q,\gamma)$.
\item[Instance] A matrix $\matrix$ over $\gf2$ with rows $\rows$ and 
columns $\columns$   representing a binary matroid~$\matroid$.
\item[Output]  $\ZtildeTutte(\matroid;q,\boldgamma)$,
where
$\boldgamma$ is the constant function  with $\boldgamma_\column = \gamma$ for every $\column\in\columns$.
\end{description}

\subsection{Standard definitions: approximation schemes and approx\-imation-preserving reductions}
\label{sec:defs}

We are interested in the complexity of \emph{approximately} solving the problem $\binTutte(q,\gamma)$.
We start with the relevant definitions. The
reader who is already familiar with the complexity of
approximate counting can skip this section. We use the presentation from~\cite{FerroPotts}. 
 
A \emph{randomised approximation scheme\/} is an algorithm for
approximately computing the value of a function~$f:\Sigma^*\rightarrow
\mathbb{R}$.
The
approximation scheme has a parameter~$\varepsilon>0$ which specifies
the error tolerance.
A \emph{randomised approximation scheme\/} for~$f$ is a
randomised algorithm that takes as input an instance $ x\in
\alphabet^{\ast }$ (e.g., for the problem $\binTutte(q,\gamma)$, the
input would be a matrix~$\matrix$ over~$\gf2$
representing a binary matroid~$\matroid$) and a rational error
tolerance $\varepsilon >0$, and outputs a rational number $z$
(a random variable depending on the ``coin tosses'' made by the algorithm)
such that, for every instance~$x$,
$ 
 \Pr \big[e^{-\epsilon} f(x)\leq z \leq e^\epsilon f(x)\big]\geq \tfrac{3}{4}$.
The randomised approximation scheme is said to be a
\emph{fully polynomial randomised approximation scheme},
or \emph{FPRAS},
if it runs in time bounded by a polynomial
in $ |x| $ and $ \epsilon^{-1} $.
As in~\cite{FerroPotts}, we say that a real number~$z$ is \emph{efficiently approximable} if there is an FPRAS for the constant function $f(x)=z$.

Our main tool for understanding the relative difficulty of
approximation counting problems is \emph{approximation-preserving reductions}.
We use
Dyer, Goldberg, Greenhill and Jerrum's notion of
approximation-preserving reduction~\cite{APred}.
Suppose that $f$ and $g$ are functions from
$\alphabet^{\ast }$ to~$\mathbb{R}$. An ``approximation-preserving
reduction'' from~$f$ to~$g$ gives a way to turn an FPRAS for~$g$
into an FPRAS for~$f$. Here is the definition. An {\it approximation-preserving reduction\/}
from $f$ to~$g$ is a randomised algorithm~$\mathcal{A}$ for
computing~$f$ using an oracle for~$g$. The algorithm~$\mathcal{A}$ takes
as input a pair $(x,\varepsilon)\in\alphabet^*\times(0,1)$, and
satisfies the following three conditions: (i)~every oracle call made
by~$\mathcal{A}$ is of the form $(w,\delta)$, where
$w\in\alphabet^*$ is an instance of~$g$, and $0<\delta<1$ is an
error bound satisfying $\delta^{-1}\leq\poly(|x|,
\varepsilon^{-1})$; (ii) the algorithm~$\mathcal{A}$ meets the
specification for being a randomised approximation scheme for~$f$
(as described above) whenever the oracle meets the specification for
being a randomised approximation scheme for~$g$; and (iii)~the
run-time of~$\mathcal{A}$ is polynomial in $|x|$ and
$\varepsilon^{-1}$.

If an approximation-preserving reduction from $f$ to~$g$
exists we write $f\APred g$, and say that {\it $f$ is AP-reducible
  to~$g$}.
Note that if $f\APred g$ and $g$ has an FPRAS then $f$ has an FPRAS\null.
(The definition of AP-reduction was chosen to make this true).
If $f\APred g$ and $g\APred f$ then we say that
{\it $f$ and $g$ are AP-interreducible}, and write $f\APeq g$.
A word of warning about terminology:
Subsequent to~\cite{APred}, the notation $\APred$ has been
used
to denote a different type of approximation-preserving reduction which applies to
optimisation problems.
We will not study optimisation problems in this paper, so hopefully this will
not cause confusion.

Dyer et al.~\cite{APred} studied counting problems in \#P and
identified three classes of counting problems that are interreducible
under approx\-imation-preserving reductions. The first class, containing the
problems that admit an FPRAS, are trivially AP-interreducible since
all the work can be embedded into the reduction (which declines to
use the oracle). The second class  is the set of problems that are
AP-interreducible with \SAT, the problem of counting
satisfying assignments to a Boolean formula in CNF\null.
Zuckerman~\cite{zuckerman}
has shown that \SAT{} cannot have an FPRAS unless
$\mathrm{RP}=\mathrm{NP}$. The same is obviously true of any problem
 to which \SAT{} is AP-reducible.  
 
The third class appears to be of intermediate complexity.
It contains   all of the counting problems
expressible in a certain logically-defined complexity class, $\RHPi$. Typical
complete problems include counting the downsets in a partially ordered
set~\cite{APred},
computing the partition function of the ferromagnetic Ising model with
local external magnetic fields~\cite{ising},
and counting the independent sets in a bipartite graph,
which is defined as follows.

\begin{description}
\item[Problem] $\BIS$.
\item[Instance] A bipartite graph $B$.
\item[Output]  The number of independent sets in $B$.
\end{description}

In \cite{APred} it was shown that \BIS\ is complete for the
logically-defined
complexity class $\mathrm{\#RH}\Pi_1$  with respect to approximation-preserving
reductions.
We conjecture that there is no FPRAS for \BIS.
 
\section{Approximating the Tutte polynomial of a binary matroid}

This section provides the proof of the following Theorem.

\begin{theorem} 
\label{thm:binmatroid}
Suppose that $q\geq 2$ and $\gamma$ are efficiently approximable.
\begin{enumerate}
\item \label{itemspecial}
If $\gamma=0$, or if $(q,\gamma)$ is one of the special points $(4,-2)$, $(2,-2)$ or $(2,-1)$,
then $\binTutte(q,\gamma)$ can be solved exactly in polynomial time.
\item \label{itemantiferropotts}
Otherwise, if  $\gamma < 0$ then there is no FPRAS for 
$\binTutte(q,\gamma)$ unless $\NP=\RP$. 
\item \label{itemnew}
If $\gamma>0$ then $\BIS\APred \binTutte(q,\gamma)$.
\end{enumerate}
\end{theorem}

Some of the parts of the theorem follow
from our earlier work in~\cite{tuttepaper} and~\cite{FerroPotts}.
The main new result is Item~(\ref{itemnew}).
Its proof follows from
\begin{enumerate}
\item the AP-reduction from $\BIS$ to  
$\uhTutte(2,1)$ from our paper~\cite{FerroPotts},
\item an AP-reduction from $\uhTutte(2,1)$
to the problem of computing the Tutte polynomial of a binary matroid, where the 
values~$\gamma_e$ depend on the input (Lemma~\ref{lem:binmatroid}), and
\item implementation of these values~$\gamma_e$ using
series-parallel extensions on binary matroids (Lemma~\ref{lem:shift}).
\end{enumerate}
The details are in the following sections.

\subsection{The Tutte polynomial of a uniform hypergraph}

We have seen one possible generalisation of the Tutte polynomial of a graph,
namely to binary matroids.  Another natural generalisation
is to hypergraphs.  The two generalisation are different, but the 
relationship between them is interesting, and and will be exploited
in one of our reductions.

It is typical to define a hypergraph as a pair $(\hypervertices,\hyperedges)$
in which $\hypervertices$ is a set of vertices, and $\hyperedges$ is a set of non-empty 
subsets of $\hypervertices$, called hyperedges.
For our work on the Tutte polynomial, it will be more convenient
to extend this definition. 
Thus, we will use the term ``hypergraph'' to refer to a pair
$(\hypervertices, \hyperedges)$ in which $\hypervertices$ is a set of vertices,
and $\hyperedges$ is a \emph{multiset} of
non-empty 
subsets of $\hypervertices$, called hyperedges.
The reason that the collection $\hyperedges$ of hyperedges is a multiset, rather than a set,
is that it is useful for the Tutte polynomial to allow ``parallel'' edges so that certain
operations, such as parallel extensions, which we shall define below, can be freely applied~\cite{Sokal05}.
A hypergraph is {\it uniform\/} if all hyperedges have the same cardinality.

Let $\hypergraph=(\hypervertices,\hyperedges)$ be a hypergraph.
The multivariate Tutte polynomial of $\hypergraph$ 
was studied (under a different name) by Grimmett~\cite{Grimmett}. A definition can be found, for
example, in~\cite{FerroPotts}. 
In this paper we will use the Potts model version.
Suppose that~$q$ is a positive integer 
and that $\boldgamma=\{\gamma_\hyperedge\}_{\hyperedge\in\hyperedges}$. 
Let 
\begin{equation}\label{eq:HyperPottsDefn}
\ZPotts(\hypergraph;q,\boldgamma) = 
\sum_{\sigma:\hypervertices\rightarrow [q]}\,
\prod_{\hyperedge\in\hyperedges}
\big(1+\gamma_\hyperedge\delta(\{\sigma(\hypervertex) \mid \hypervertex\in \hyperedge\})\big),
\end{equation}
where $[q]=\{0,\ldots,q-1\}$ is a set of $q$~spins or colours,
and $\delta(S)$ is~$1$ if its argument is a singleton and 0 otherwise.
Identity (\ref{eq:HyperPottsDefn}) extends the Tutte polynomial from 
graphs to hypergraphs, but only for positive integer~$q$. 
It is possible to provide a formulation for general $q$ along the 
lines of~(\ref{eq:randClust}), but this is not needed in what follows.

We consider the following computational problem,
\begin{description}
\item[Problem] $\uhTutte(q,\gamma)$.
\item[Instance]  A \emph{uniform} hypergraph $\hypergraph=(\hypervertices,\hyperedges)$.
\item[Output]  $\ZPotts(\hypergraph;q,\boldgamma)$,
where $\boldgamma$ is the constant function  
 with $\boldgamma_\hyperedge = \gamma$ for every $\hyperedge\in\hyperedges$.
\end{description}

The will use the following lemma, which is an easy consequence
of the results of~\cite{FerroPotts}.
 
\begin{lemma}  
\label{lem:BIS}
$\BIS\APred\uhTutte(2,1)$.
\end{lemma}

\begin{proof}
This follows from Observation~2, Lemma~14 and
Lemma~15 of~\cite{FerroPotts}.
We note that \cite{FerroPotts} stated a more general definition of ``hypergraph'' in 
which hyperedges were taken to be multisets, rather than sets. Nevertheless,
the construction in Lemma~15 actually produces a hypergraph that conforms to the definition that
we use here.
\end{proof}

\subsection{A Potts model characterisation}

Just as we used the Potts model version of the
multivariate Tutte polynomial of a hypergraph,
it will be helpful to have a representation of $\ZtildeTutte(\matroid;q,\boldgamma)$ in terms of the
(multivariate) partition function of the Potts model. See~\cite[Theorem 3.1]{Sokal05}.
Let $\matroid$ be a matroid represented by a 
matrix
$\matrix$ over $\gfq$ with rows $\rows$ and columns $\columns$.
For every column~$\column\in\columns$, let
$\lineq_\column$
be the linear equation 
$\sum_{\row\in \rows} \matrix_{\row,\column}\sigma(\row) = 0$
where the arithmetic is in~$\gfq$. 
The Potts partition function
of $\matroid$ is
defined as follows:
\begin{equation}\label{eq:Pottspfdefn}
\ZPotts(\matroid;q,\boldgamma) = 
\sum_{\sigma:\rows\rightarrow[q]}\,
\prod_{\column\in\columns}
(1+\gamma_\column\delta_\column(\sigma)),
\end{equation}
where  
$$\delta_\column(\sigma)
=\begin{cases}
1, & \text{if $\sigma$ satisfies $\lineq_\column$,} \\
0, & \text{otherwise.}
\end{cases}
$$
Also, let
$\ZIsing(\matroid;\boldgamma)$ be a synonym for $\ZPotts(\matroid;2,\boldgamma)$. The Ising
model is the special case $q=2$ of the Potts model.
Note that
\begin{align*}
\ZIsing(\matroid;\boldgamma) & =  
\sum_{\sigma:\rows\rightarrow \{0,1\}}
\sum_{\subsetcols\subseteq\columns}\,
\prod_{\column\in\subsetcols}
\gamma_\column\delta_\column(\sigma)\\
& = \sum_{\subsetcols\subseteq\columns}\,
\gamma_\subsetcols
\sum_{\sigma:\rows\rightarrow \{0,1\}}\,
\prod_{\column\in\subsetcols} \delta_\column(\sigma),
\end{align*}
where $\gamma_\subsetcols=\prod_{\column\in\subsetcols} \gamma_\column$.
The number of configurations~$\sigma:\rows\rightarrow \{0,1\}$ 
for which $\prod_{\column\in\subsetcols} \delta_\column(\sigma)=1$
is the number of solutions to the system of linear equations 
$\lineq_\subsetcols = 
\left\{ \lineq_\column  \mid \column \in \subsetcols
\right\}$, which we denote $\#\lineq_\subsetcols$.
Thus, 
\begin{align}
\label{eq:fkising}
\ZIsing(\matroid;\boldgamma) &= 
 \sum_{\subsetcols\subseteq\columns}
\gamma_\subsetcols\,\#\lineq_\subsetcols\\ \nonumber
&=\sum_{\subsetcols\subseteq\columns}
\gamma_\subsetcols\,
2^{\numrows-\rank_{\matroid}(\subsetcols)}\\ \nonumber
  &= 2^{|V|}\,\ZtildeTutte(\matroid;2,\boldgamma).
\end{align}

It is interesting to compare definitions (\ref{eq:HyperPottsDefn}) and~(\ref{eq:Pottspfdefn})
in the case $q=2$ to see how they both arise as natural generalisations 
of the Ising partition function of  a graph.
In the classical Ising model on a graph, each edge $(u,v)$ contributes a factor 
depending on whether $\sigma(u)$ and $\sigma(v)$ are equal.  If we think of this condition as 
asserting that the edge $(u,v)$ is monochromatic, then the extension~(\ref{eq:HyperPottsDefn})
to hypergraphs is immediate.  On the other hand, we can equally think of the same condition
as asserting $\sigma(u)+\sigma(v)=0\pmod2$, 
which leads us naturally to
definition~(\ref{eq:Pottspfdefn}) 
for binary matroids.

\subsection{Reduction from \textmd{$\uhTutte(2,\gamma)$}}

Consider the following computational problem, which is similar to~$\binTutte(q,\gamma)$
except that the weight~$\gamma$ is part of the input.

\begin{description}
\item[Problem] $\varybinTutte(q)$.
\item[Instance] A matrix $\matrix$ over $\gf2$ with rows $\rows$ and 
columns $\columns$   representing a binary matroid~$\matroid$.
A positive integer~$\myN$, given in unary.
\item[Output]  $\ZtildeTutte(\matroid;q,\boldgamma)$,
where
$\boldgamma$ is the constant function  with $\boldgamma_\column =  2^{2/\myN}-1$ for every $\column\in\columns$.
\end{description}

\begin{lemma}
$\uhTutte(2,1) \APred \varybinTutte(2)$.  
\label{lem:binmatroid}
\end{lemma}

\begin{proof}

Let  
$\hypergraph=(\hypervertices,\hyperedges)$ be a
$\terms$-uniform hypergraph, an instance of $\uhTutte(2,1)$.
Without loss of generality, assume $\terms>2$ since the result is immediate for $t=2$
(since a 2-uniform hypergraph is a graph, and a graphic matroid is binary).
Let $n=|\hypervertices|$ and $m=|\hyperedges|$ and assume that these are sufficiently large.
Let $\boldone$ be the constant function  
which maps   every $\hyperedge\in\hyperedges$ to~$1$.
By definition,
$$\ZPotts(\hypergraph;2,\boldone) = 
\sum_{\sigma:\hypervertices\rightarrow \{0,1\}}
2^{\mono(\sigma)},$$
where $\mono(\sigma)$ denotes the number of hyperedges $\hyperedge\in\hyperedges$
that are monochromatic in configuration~$\sigma$.

Let $\epsilon$ be the desired accuracy of the AP-reduction and
let $\delta = \epsilon/(m \ln 2)$. Let~$\myN$ be any
positive integer satisfying
$$\myN \geq
\frac{6m^2(n+\ln(16m))}{\epsilon^2}.
$$
We will construct a  
$n \times N m$
matrix~$\matrix$ 
so that~$\myN$ and~$\matrix$ constitute an input to~$\varybinTutte(2)$.
The rows of~$\matrix$ correspond to the elements of $\hypervertices$.
The matrix contains $N$ columns, $f_1, \ldots, f_\myN$, for
each hyperedge $f\in\hyperedges$.
When we construct the matrix, we choose each of
these columns to be the indicator vector 
for   an even-sized subset of~$f$, chosen independently and uniformly at random.

Given our construction, it is easy to see that
a configuration~$\sigma$ which is monochromatic on~$f$
will satisfy the equations corresponding to all~$\myN$ columns $f_1,\ldots,f_\myN$.
For this, it is important that the random subsets of~$f$ corresponding to these columns have
even size since the relevant equation $\lineq_{f_j}$ is   
$\sum_{\row\in \hypervertices} \matrix_{\row,f_j}\sigma(\row) = 0 \bmod 2$.

Suppose that a configuration~$\sigma$ is not monochromatic
on a hyperedge~$\hyperedge$ and that it assigns~$\ell$ elements of~$f$ to spin~$1$
and~$k$ elements of~$f$ to spin~$0$ for positive integers~$\ell$ and~$k$.
Note that $\ell+k=t$ since the hypergraph is $t$-uniform.
The number of even-sized subsets of~$f$ is~$2^{t-1}$
and the number of even-sized  subsets for which~$\sigma$ restricted to that
subset has an even number of 1s
is~$2^{\ell-1}2^{k-1} = 2^{t-2}$. 
Thus, the  probability
that the equation associated with a column~$f_i$ is satisfied is~$1/2$.
So, by a Chernoff bound, the number of columns~$f_i$
with satisfied equations is with high probability in the range
$[{(\myN/2)(1-\delta)},{(\myN/2)(1+\delta)}]$.
Specifically, the failure probability for this event is 
at most
$$
2\exp(- {\delta^2} \myN/6) \leq 2\exp(-n-\ln 16m)\leq  2\times2^{-n}\times\frac1{16m}=
\frac1{8m2^n}
$$
\cite[Cor.~4.6]{MitzEli}.  From the union bound --- ranging over events
indexed by the $2^n$ choices for $\sigma$ and $m$ choices of column~$f$ ---
we conclude that, with probability at least~$7/8$, the following is true.
For every configuration~$\sigma$, for every hyperedge $f\in\hyperedges$ on which $\sigma$
is not monochromatic, the number of columns in $\{f_1,\ldots,f_N\}$ with
equations satisfied by~$\sigma$ is in the
range $[{(\myN/2)(1-\delta)},{(\myN/2)(1+\delta)}]$.

Let~$\matroid$ be the binary matroid represented by~$\matrix$. Let $y=2^{2/\myN}$
and $\gamma = y-1$. Let $\boldgamma$ be the constant function which maps every element of the ground
set of~$\matroid$ to~$\gamma$ as in the definition of \varybinTutte.
To complete the verification of our reduction, we need to show that 
 $\ZPotts(\hypergraph;2,\boldone)$ may be easily computed given 
 $\ZtildeTutte(\matroid,2,\boldgamma)$;  note that the latter quantity 
 is equal to $2^{-|V|}\ZIsing(\matroid;\gamma)$ by (\ref{eq:fkising}).
Now the contribution of a configuration~$\sigma$
to the quantity 
$2^m \ZPotts(\hypergraph;2,\boldone)$
is $2^m 2^{\mathrm{mono}(\sigma)}$.
Let $\Psi_\sigma$ be the contribution of~$\sigma$
to  
$\ZIsing(\matroid;\boldgamma)$.
Then
\begin{align*} \Psi_\sigma &\leq
  {y}^{\,\myN \,\mathrm{mono}(\sigma)}
{y}^{(m - \mathrm{mono}(\sigma)) (\myN/2)(1+\delta)}\\
&\leq
{y}^{m\myN(1+\delta)/2}
  {y}^{(\myN/2) \mathrm{mono}(\sigma)}
= 2^{m \delta}  2^m 2^{\mathrm{mono}(\sigma)}.
\end{align*}
Also, we get a similar lower bound.
\begin{align*}  
\Psi_\sigma &\geq
 {y}^{\,\myN \,\mathrm{mono}(\sigma)}
{y}^{(m - \mathrm{mono}(\sigma)) (\myN/2)(1-\delta)}\\
 &\geq
{y}^{m\myN(1-\delta)/2}
 {y}^{(\myN/2) \mathrm{mono}(\sigma)}
 = 2^{-m \delta}  2^m 2^{\mathrm{mono}(\sigma)}. 
\end{align*}
The reduction has the desired accuracy for an AP-reduction, since 
$2^{m \delta} = e^{\epsilon}$.
Assume that the oracle call to $\varybinTutte(2)$ is powered to have failure probability~$\frac18$.
Then the overall failure probability is bounded by $\frac14$, being 
the sum of $\frac18$ from the randomised nature of the reduction itself,
and $\frac18$ from the single oracle call.
 \end{proof}

\subsection{Series-parallel extensions of binary matroids}
\label{sec:seriesparallel}

The standard method for reducing the problem of evaluating the Tutte polynomial
with some weight~$\gamma'$ to a Tutte-polynomial evaluation problem with a different
weight~$\gamma$ is to ``implement'' the edge weight~$\gamma'$ using
series and parallel extensions of weight-$\gamma$ edges. 
See~\cite[Section 2.3]{Sokal05} and~\cite[Section 10]{FerroPotts}.
Series and parallel extensions of matroids are generalisations of the 
stretchings and thickenings of graphs used by Jaeger et al.~\cite{JVW90}.
Sokal~\cite[Section 2.3]{Sokal05} has given the details, both for graphs and for general matroids.
It is fairly
easy to show that these extensions can be done within the class of binary matroids.
We do this here.

\begin{lemma} [parallel extension for binary matroids]
\label{lem:thicken}
Let $\matroid$ be a binary  matroid represented by the matrix~$\matrix$
with rows~$\rows$ and columns~$\columns$.
Let $\boldgamma=\{\gamma_\column\}_{\column\in\columns}$.
Let $c$ be any column in~$\columns$.
Suppose $\gamma_1>0$ and $\gamma_2>0$ satisfy
\begin{equation}
\label{eq:parallel}
1+\gamma_c  = (1+\gamma_1)(1 + \gamma_2).
\end{equation}
Let~$e'$ be a new column.
Define $\boldgamma'=\{\gamma'_\column\}_{\column\in\columns \cup \{e'\}}$
as follows. 
Let $\gamma'_{e'}=\gamma_2$ and $\gamma'_{c} = \gamma_1$. 
For every other column $\column\in \columns$, 
let $\gamma'_\column = \gamma_\column$.
There is a binary matroid $\matroid'$ represented by a matrix~$\matrix'$ with $\numrows$ rows
and columns $\columns\cup\{e'\}$
for which $\ZtildeTutte(\matroid;q,\boldgamma) = \ZtildeTutte(\matroid';q,\boldgamma')$.
\end{lemma}
 
 \begin{proof}
 Let $M'$ be the matrix constructed from~$M$  by making column~$e'$ a copy of column~$c$.
 Let $\matroid'$ be the matroid represented by~$M'$.
Note that, for any $A\subseteq E$, $r_{\calM}(A) = r_{\calM'}(A)$. Thus, from the
definitions in Section~\ref{sec:matroid}, $\calM = \calM'\backslash e'$. 
The result now follows from \cite[(4.22)]{Sokal05}, provided
that we can show that either (i) $ c$ and $e'$ form a two-element circuit of $\calM'$, or
(ii)  $c$ and $e'$ are both loops of $\calM'$.
(These are the two side-conditions for the application of
\cite[(4.22)]{Sokal05}.)
 Now, If $c$ is not the all-zero vector,
 then $c$ and $e'$ do form a two-element circuit (minimal dependent set) of $\calM'$
 since $r_{\calM'}(\{c,e'\}) = r_{\calM'}(\{c\})= r_{\calM'}(\{e'\})=1$. 
 Otherwise, $r_{\calM'}(\{c\}) = r_{\calM'}(\{e'\})=0$, so
 both $c$ and $e'$ are loops of $\calM'$.   
 \end{proof}

\begin{lemma}[series extension for binary matroids]
\label{lem:stretch}
Let $\matroid$ be a binary  matroid represented by the matrix~$\matrix$
with rows~$\rows$ and columns~$\columns$. 
Let $\boldgamma=\{\gamma_\column\}_{\column\in\columns}$.
Let $c$ be any column in~$\columns$.
Assume $q\not=0$, and suppose $\gamma_1>0$ and $\gamma_2>0$ satisfy
\begin{equation}
\label{eq:series}
\left(1+\frac{q}{\gamma_c}\right)
= \left(1+\frac{q}{\gamma_1}\right)\left(1+\frac{q}{\gamma_2}\right)
\end{equation}
Let~$e'$ be a new column.
Define $\boldgamma'=\{\gamma'_\column\}_{\column\in\columns \cup \{e'\}}$
as follows. 
Let $\gamma'_{e'}=\gamma_2$ and $\gamma'_{c} = \gamma_1$. 
For every other column $\column\in \columns$, 
let $\gamma'_\column = \gamma_\column$.
There is a binary matroid $\matroid'$ represented by a matrix~$\matrix'$ with $\numrows+1$ rows
and columns $\columns\cup\{e'\}$
for which $(1+\gamma_1/q + \gamma_2/q)\, \ZtildeTutte(\matroid;q,\boldgamma) = \ZtildeTutte(\matroid';q,\boldgamma')$.
\end{lemma}
 
\begin{proof}
Let $M'$ be the matrix constructed from $M$ by adding a new column~$e'$ and a new row $r'$.
The new row has ones in column~$c$ and column~$e'$ only. There are no other ones in column~$e'$.
Let $\matroid'$ be the matroid represented by~$M'$.  
We first show that, for any $A\subseteq E$,
\begin{equation}\label{eqtemp}
r_{\calM}(A) = r_{\calM'}(A\cup e')-1.
\end{equation}
Equation (\ref{eqtemp}) can be verified by checking three cases
\begin{itemize}
\item $A\subseteq E-\{c\}$,
\item $c\in A$ and $r_{\calM}(A) = r_{\calM}(A-\{c\})$, and
\item $c\in A$ and $r_{\calM}(A) = r_{\calM}(A-\{c\})+1$.
\end{itemize}
(The conditions on the rank function of a matroid guarantee that these are the only three cases.)
Now note that $r_{\calM'}(\{e'\})=1$,
so (\ref{eqtemp}) implies
that $r_{\calM}(A) = r_{\calM'}(A\cup \{e'\})-r_{\calM'}(\{e'\})$. We conclude (from the definitions in Section~\ref{sec:matroid}) that
$\calM = \calM' / e'$.
The result now follows from \cite[(4.28)]{Sokal05}, provided that we can show either
(i) $c$ and $e'$ form a cocircuit of $\calM'$, or
(ii) $c$ and $e'$ are both coloops of $\calM'$.
(These are the two side-conditions for the application of
\cite[(4.28)]{Sokal05}.)

Suppose first 
that $c$ is a coloop of $\calM$ (i.e., an element that is present in every basis of~$\calM$).
Consider any independent set $A\subseteq E-\{c\}$ of $\calM'$.
Since $c$ is linearly independent of the columns in~$A$ in $M$,
both $c$ and $e'$ are linearly independent of the columns in~$A$ in $M'$.
Thus, both $c$ and $e'$ are coloops of $\calM'$ and we have (ii).

Finally, suppose that $c$ is not a coloop of $\calM$. Our goal is to prove (i).
First, Any independent set in~$\matroid'$ including neither $c$ nor~$e'$ can be extended 
to a larger independent set by adding either one of $c$ or~$e'$, so every basis of $\matroid'$
intersects $\{c, e'\}$.  
However,  since $c$ is not
a coloop of $\calM$, there is a basis $B$ of~$\matroid$ that does not include~$c$.  
Then $B \cup\{c\}$ and $B\cup\{e'\}$ are both
bases of~$\matroid'$ 
and so $\{c, e'\}$ is a minimal set that intersects every basis of $\calM'$. Thus, we have (i).
\end{proof}

\subsection{Implementing variable weights}

\begin{lemma}
\label{lem:shift}
Suppose that $\gamma>0$ is efficiently approximable.
Then
$\varybinTutte(2) \APred \binTutte(2,\gamma)$.
\end{lemma}

\begin{proof}
Let $C_\gamma$ be a sufficiently large function of the parameter~$\gamma$.
The exact computation of~$C_\gamma$ is from~\cite{FerroPotts}. This will be explained below.

Let~$\matrix$ and~$\myN$ be an instance of $\varybinTutte(2)$. 
Let~$\matroid$ be the matroid represented by~$\matrix$.
Suppose that~$\matrix$ has
$n$~rows and $m$~columns and that the product $\myN m$ is sufficiently large with respect to
the constant~$C_\gamma$.
Let $\gamma' = 2^{2/\myN}-1$.
Let $\boldgamma'$ be the constant function which maps every 
ground set element of~$\matroid$ 
to $\gamma'$.

The proof is based on the proof of~\cite[Lemma 17]{FerroPotts}. 
Let $\epsilon$ be the desired accuracy in the approximation-preserving reduction.
Let 
$$\chi = \frac{\epsilon^2}
{4 C_\gamma m^2 N}.$$
 Let  
$\hat\gamma$ be a rational in the range
$e^{-\chi} \gamma \leq \hat \gamma \leq e^{\chi} \gamma$.
Since   $\gamma$  is efficiently approximable,  the amount of time
that it takes to compute   $\hat \gamma$
is at most a polynomial in $m$, $N$ and $\epsilon^{-1}$.

The idea of the proof is to show how to use 
series and parallel extensions of weight-$\hat\gamma$ elements to 
implement weight ${\gamma}^*$ 
satisfying
\begin{equation}
\label{eq:finalfinal}
e^{- \chi} {\gamma'} \leq  {\gamma}^* \leq 
e^{ \chi} \gamma'.
\end{equation}
Let $\boldgamma^*$ be the 
constant function which maps every 
ground set element of~$\matroid$
to ${\gamma}^*$.
The definition of $\ZtildeTutte$ and the fact that $\chi \leq \epsilon /(4 m)$ imply that
$$
e^{-\varepsilon /4}
\ZtildeTutte(\matroid;2,\boldgamma') \leq
\ZtildeTutte(\matroid;2,\boldgamma^*) \leq 
e^{\varepsilon /4}
\ZtildeTutte(\matroid; 2, \boldgamma').
$$

Let $\hat\boldgamma$ be the 
constant function which maps every ground set element to
$\hat\gamma$.
We can think of our implementations as constructing a binary matroid $\widehat{\matroid}$
such that 
$\ZtildeTutte(\matroid;2,\boldgamma^*)$ is equal to 
the product of $\ZtildeTutte(\widehat{\matroid};2,\hat \boldgamma)$
and an easily-computed function of~$\hat \gamma$.
This easily-computed
function arises from  the extra factor 
$(1+\gamma_1/q+\gamma_2/q)$ in Lemma~\ref{lem:stretch}.
We will ensure that the 
matroid~$\widehat{\matroid}$ has at most~$C_\gamma m^2 N/\epsilon$
ground set elements. 
To finish, we note, using the definition of~$\ZtildeTutte$ and the definition of~$\chi$,  that 
$$e^{-\varepsilon/4 }
\ZtildeTutte(\widehat{\matroid};2,\boldgamma) \leq
\ZtildeTutte(\widehat{\matroid};2,\hat \boldgamma) \leq 
e^{\varepsilon /4}
\ZtildeTutte(\widehat{\matroid}; 2, \boldgamma),$$
where $\boldgamma$ is the constant weight function which assigns every element weight $\gamma$. 
We finish the approximation of 
$ \ZtildeTutte(\widehat{\matroid};2,\hat \boldgamma)$
by using the oracle to approximate
$\ZtildeTutte(\widehat{\matroid};2,\boldgamma) $ using accuracy parameter $\delta = \epsilon/2$.

It remains to show how to do the implementation.
Take
$$\pi = \frac{\chi}{2}(2^{2/\myN}-1) 
= \frac{\chi}{2} \gamma'
\leq \gamma'(1-e^{-\chi}).$$
The proof of~\cite[Lemma 17]{FerroPotts} 
shows how to use series and parallel extensions of weight~$\hat\gamma$ elements
(from Lemmas~\ref{lem:thicken} and~\ref{lem:stretch})
to implement a weight
 ${\gamma}^*$
which satisfies ${\gamma'} - \pi \leq {\gamma}^* \leq {\gamma'}$.
This ensures that Equation~(\ref{eq:finalfinal}) holds. 
The series and parallel extensions
in the implementation of~$\gamma^*$ introduce at most~$C_\gamma \log(\pi^{-1})$ ground set elements,
where $C_\gamma$ is some quantity depending on~$\gamma$ but not on~$n$, $m$, or~$\myN$.
Note that $2^{2/\myN}-1 \geq 2 \ln(2)/\myN$ so
$\pi^{-1}\leq 2N/\chi$.
Thus, $m C_\gamma  \log(\pi^{-1}) \leq C_\gamma m^2 N/\epsilon$
and the matroid~$\widehat{\matroid}$ has 
at most $C_\gamma m^2 N/\epsilon$ ground set elements, as required above.
\end{proof}

\subsection{The Proof of  Theorem~\ref{thm:binmatroid}}

\begin{proof}[Proof of Theorem~\ref{thm:binmatroid}]
Item~(\ref{itemspecial}) comes from~\cite{JVW90}.
Item~(\ref{itemantiferropotts}) follows quickly
from~\cite{tuttepaper} together with an application of matroid 
duality;  here are the details.
When $q>2$, item~(\ref{itemantiferropotts}) follows from the
corresponding hardness result for graphs, since binary matroids 
generalise graphic matroids.  
The same is true when $q=2$ and $\gamma>-2$.
The complementary case $q=2$ and $\gamma<-2$ then follows by matroid duality,
as can be seen by combining the following observations:
Binary matroids 
are closed under duality \cite[2.2.9]{oxley}
and the representation of the dual of a binary matroid can be constructed efficiently
\cite[2.2.8]{oxley}.
Also, if $\matroid^*$ is the dual of~$\matroid$ then
$\ZtildeTutte(\matroid;q,\gamma)$ is an easily-computed
multiple of
$\ZtildeTutte(\matroid^*,q,q/\gamma)$ --- see \cite[(4.14)]{Sokal05}.
Finally for $q=2$, if $\gamma<-2$ then $-2<q/\gamma < 0$.

 Item~(\ref{itemnew}) follows from \cite[Theorem 1]{FerroPotts} for $q>2$ 
since the cycle matroid of a graph is a binary matroid. 
For $q=2$ it follows from Lemmas~\ref{lem:BIS}, Lemma~\ref{lem:binmatroid} and~\ref{lem:shift}.
\end{proof}

\section{The weight enumerator of a binary linear code}
\label{sec:we}

Given a \emph{generating matrix}
$M$ over $\gf2$ with $r$ linearly independent rows and $c$~columns,
a \emph{code word}~$w$ is any
vector in the linear subspace~$\Upsilon$ generated by the rows of~$M$.
For any real number~$\lambda$, the \emph{weight
enumerator} of the code is given by $W_M(\lambda)=\sum_{w\in
\Upsilon}\lambda^{\|w\|}$, where $\|w\|$ is the number of non-zero
entries 
in~$w$. 
We consider the following computational problem, parameterised
by~$\lambda$.
\begin{description}
\item[Problem] $\WE(\lambda)$.
\item[Instance]  A generating matrix
$M$ over GF($2$).
\item[Output]  $W_M(\lambda)$.
\end{description}

It is well-known (see below) that the weight enumerator of a binary linear code is 
a special case of the Tutte polynomial of a binary matroid.
Thus, Theorem~\ref{thm:binmatroid}
has the following corollary.
\begin{corollary}
\label{cor:we}
Suppose that $\lambda$ is efficiently approximable.
\begin{enumerate}
\item If $\lambda\in \{-1,0,1\}$ then $\WE(\lambda)$ is solvable in polynomial time.
\item If $|\lambda|>1$ then there there is no FPRAS for $\WE(\lambda)$ unless $\NP=\RP$.
\item If $\lambda\in (-1,0)$ then there there is no FPRAS for $\WE(\lambda)$ unless $\NP=\RP$.
\item If $\lambda\in(0,1)$ then $\BIS \APred  \WE(\lambda)$.\end{enumerate}
\end{corollary}

Corollary~\ref{cor:we} follows immediately from Theorem~\ref{thm:binmatroid}
and from  Lemma~\ref{lem:WE} below, which is the $q=2$ case of a result 
of Greene~\cite[Corollary 4.5]{Greene}. See also Cameron~\cite[Theorem 4.1]{cameron},
but note that both authors employ a different parameterisation of the Tutte polynomial.
We provide a short proof here for completeness, since 
we have already done almost all of the necessary work.
 
\begin{lemma} [Greene]
\label{lem:WE}
Let $M$ be a generating matrix over $\gf2$ with rows~$\rows$ and columns~$\columns$.
Let $\matroid$ be the binary matroid represented by~$M$. 
Let $\lambda$ be any non-zero real number and let $\gamma = 1/\lambda-1$.
Let $\boldgamma$ be the constant function  with $\boldgamma_\column = \gamma$ for every 
column~$\column$ of~$M$.
Then 
$$W_M(\lambda) = { \lambda}^{|\columns|}2^{|V|}\ZtildeTutte(\matroid;2,\boldgamma).$$ \end{lemma}

\begin{proof} 
As in the proof of Lemma~\ref{lem:binmatroid},
we have, from equations (\ref{eq:Pottspfdefn}) and~(\ref{eq:fkising}),
\begin{equation}\label{Isingexpression}
2^{|V|}\ZtildeTutte(\matroid;2,\boldgamma) = \ZIsing(\matroid;\boldgamma)=
\sum_{\sigma:\rows\rightarrow \{0,1\}}
{(1+\gamma)^{\goodcol(\sigma)}}=
\sum_{\sigma:\rows\rightarrow \{0,1\}}
{\lambda^{-\goodcol(\sigma)}},
\end{equation}
where $\goodcol(\sigma)$ denotes the number of columns $\column\in\columns$
such that 
$$\sum_{\row\in \rows} M_{\row,\column}\sigma(\row) = 0\pmod2,$$
or, using the notation introduced earlier, $\delta_\column(\sigma)=1$.
Similarly, 
\begin{equation}\label{WEexpression}
W_M(\lambda)
=  
\sum_{\sigma:\rows\rightarrow \{0,1\}}\,
\prod_{e\in\columns} \lambda^{1-\delta_\column(\sigma)}
=\sum_{\sigma:\rows\rightarrow \{0,1\}}\lambda^{|E|-\goodcol(\sigma)},
\end{equation}
where, corresponding to~$\sigma$, 
the codeword~$w$ is the sum of the rows~$i$ with $\sigma(i)=1$,
so $1-\delta_\column(\sigma)$ is the 
bit in position~$\column$ of the code-word.
The result follows by comparing (\ref{Isingexpression}) and (\ref{WEexpression}).
\end{proof}

\section{The cycle index polynomial}
\label{sec:ci}
 
Let $\group$ be a group of permutations of $\{1,\ldots,\points\}$.
Each permutation $\perm\in\group$ decomposes the
set $\{1,\ldots,\points\}$ into a collection of cycles. $\cycles(\perm)$ denotes
the number of cycles in this decomposition. 
The single-variable \emph{cycle index polynomial} of~$\group$ is 
defined as follows.
$$\ZCycleIndex(\group;\ciparam) = \frac{1}{|\group|}
\sum_{\perm \in\group} \ciparam^{\cycles(\perm)}.$$

When~$\ciparam$ is a positive integer, $\ZCycleIndex(\group;\ciparam)$
counts the number of \emph{orbits} (or equivalence classes)
when strings from a size-$\ciparam$ alphabet are
operated on by permutations in~$\group$ (which permute the character positions in the strings).
Two strings are in the same equivalence class if there is a permutation in $\group$ which maps one into the other.
For example, when $x=2$, $\nu=3$, and $\group$ is the symmetric
group (on $3$ elements), the elements
of $\group$ are the identity permutation (which has $3$~cycles), 
the three transpositions
$(1\>2)$, $(1\>3)$ and $(2\>3)$ (each of which has $2$~cycles),
and the singleton cycles $(1\>2 \> 3)$ and $(1 \> 3 \> 2)$. Thus
$$
\ZCycleIndex(\group;2)
=\frac{1}{6} 
\left( 2^3 + 3 \times 2^2 + 2\times 2^1\right)
=4.$$ Thus, there are four orbits (namely
the orbits of the strings $000$, $001$, $011$, and $111$).
For more details, see \cite{DeBruijn,HP, bccsurvey}.
We consider the following problem, in which the parameter~$\ciparam$ is a positive real number.
\begin{description}
\item[Problem] $\CycleIndex(\ciparam)$
\item[Instance]  
A set of generators for a permutation group $\group$
\item[Output]  $\ZCycleIndex(\group;\ciparam)$.
\end{description}

We showed \cite[Theorem 4]{CIpaper} that if $\ciparam$ is not an integer then there is no
FPRAS for $\CycleIndex(\ciparam)$ unless $\RP=\NP$.
In fact, it is NP-hard to approximate
$\ZCycleIndex(\group;\ciparam)$ within any polynomial factor.
However, our technique from \cite[Theorem 4]{CIpaper} does not say
anything about the difficulty of the problem in the more interesting case
when $\ciparam$ is an integer.
We raised this question in~\cite{CIpaper} but were unable to resolve it (until the present paper).
Note that it is easy to compute $\ZCycleIndex(\group;1)$ exactly in polynomial time.
Corollary~\ref{cor:we} has the following consequence.

\begin{corollary}
\label{cor:ci} 
Suppose that $\ciparam>0$ is efficiently approximable. Then
\begin{enumerate} 
\item \label{aa} If $\ciparam=1$ then $\CycleIndex(\ciparam)$ can be solved exactly in polynomial time.
\item \label{bb} If $\ciparam$ is not an integer then there is no FPRAS for $\CycleIndex(\ciparam)$ 
unless $\NP=\RP$.
\item \label{cc} If $\ciparam>1$  is a positive integer  then
$\BIS \APred \CycleIndex(\ciparam)$ .
\end{enumerate}
\end{corollary}
\begin{proof}
 
Items (\ref{aa}) and (\ref{bb}) are from \cite[Theorem 4]{CIpaper}. We 
now prove item (\ref{cc}).
Let $\lambda = \ciparam^{-1}$.
Note that $\lambda \in (0,1)$.
We know from  Corollary~\ref{cor:we}
that $\BIS \APred  \WE(\lambda)$. To finish, we
show that  $\WE(\lambda) \APred \CycleIndex(\ciparam)$.
 
The reduction is straightforward.  Let $M$ be a generating matrix 
(an instance of $\WE(\lambda)$ with $r$ linearly independent rows and $c$~columns). 
Let~$\Upsilon$ be the subspace generated by the rows of~$M$ (this is the set of code words). 
Let $M_{i,*}$ denote row~$i$ of~$M$.
For each binary string $m\in \{0,1\}^r$,
let $w(m)$ be the vector $\sum_{i=1}^r m_i M_{i,*}$ (where arithmetic is over GF($2$))
and let $k(m)$ be the number of $1$'s in this vector.
Since the rows of~$M$ are linearly independent, 
each element of $\Upsilon$ is uniquely expressible as $w(m)$ with $m\in\{0,1\}^r$,
so 
\begin{equation} W_M(\lambda) = \sum_{m\in \{0,1\}^r} \lambda^{k(m)}.
\label{thisWM}
\end{equation}
 
 Let 
 $\points=2 c$. 
 Our objective will be to construct a group $\group$ 
 of permutations of $\{1,\ldots,\nu\}$ so that
 $W_M(\lambda)$ can be computed from $\ZCycleIndex(\group;\ciparam)$.
 For $i\in[r]$, let $g_i$ be the permutation of $\{1,\ldots,\nu\}$ 
defined as follows: For all $j\in[c]$, elements~$2j-1$ and $2j$ are mapped to each other by $g_i$ if
$M_{i,j}=1$ and each of these is mapped to itself by~$g_i$ if $M_{i,j}=0$. 
Let~$\group$ be the group of permutations of $\{1,\ldots,\nu\}$ generated by $g_1,\ldots,g_r$. 

For each binary string $m\in \{0,1\}^r$, let $g(m)$ be the
permutation $g_1^{m_1} \cdots g_r^{m_r}$,
where $g_i^{\ell}$ denotes the composition of $\ell$ copies of the generator~$g_i$
so $g_i^0$ is the identity permutation.
Note that, for each $j\in [c]$, elements $2j-1$ and $2j$ are swapped by $g(m)$
iff  
$\sum_{i=1}^r m_i M_{i,j}=1 \bmod 2$.
Thus, the number of $j$ for which there is a swap is $k(m)$ so
$\cycles(g(m)) = \nu - k(m)$.  

We will show that each permutation $g\in \group$ can be written as $g(m)$
for exactly one $m\in \{0,1\}^r$.
First, suppose that $g=g_{i_1} \cdots g_{i_\ell}$ for some $i_1,\ldots,i_\ell \in \{1,\ldots,r\}^\ell$.
Since the generators $g_1,\ldots,g_r$ commute, we can  
re-order so that
$i_1,\ldots,i_\ell$ are monotonically non-decreasing.
Then, since the generators all have order~$2$, we can 
cancel factors that are the identity permutation, making 
$i_1,\ldots,i_\ell$ distinct.
Thus, $g$ can be written as $g(m)$ for \emph{at least} one $m\in\{0,1\}^r$.
To see that $m$ is unique, suppose that
$g(m)=g(m')$ for $m\neq m'$.
Then $g_1^{m_1} \cdots g_r^{m_r} = g_1^{m'_1} \cdots g_r^{m'_r}$. 
Thus, for each $j\in [c]$,
the number of permutations in 
$\{ g_i \mid m_i=1\}$ which swap element $2j-1$ with element $2j$
has the same
parity as the number of permutations in
$\{ g_i \mid m'_i=1\}$ which swap element $2j-1$ with element $2j$.
Thus,
 $$\sum_{i=1}^r m_i M_{i,*} = \sum_{i=1}^r m'_i M_{i,*}\bmod 2$$
 Since the rows of $M$ are linearly independent, we conclude that $m=m'$.
Thus, we have proved that
$\group = \{g(m) \mid m \in \{0,1\}^r\}$
so, from the definition of $\ZCycleIndex(\group;\ciparam) $,
$$ 
\label{thisCycleIndex}
|\group| \> \ZCycleIndex(\group;\ciparam) = \sum_{m\in \{0,1\}^r} \ciparam^{\nu - k(m)}
= \ciparam^{\nu} \sum_{m\in \{0,1\}^r}   \lambda^{ k(m)}
= \ciparam^{\nu} W_M(\lambda),
$$
where the last equality uses Equation~(\ref{thisWM}).
\end{proof}

\bibliographystyle{plain}
\bibliography{mybibfile}
 
\end{document}